\documentclass[a4paper]{article}
\pdfoutput=1
\usepackage[margin=1in]{geometry}

\usepackage{graphicx}
\usepackage{comment}
\usepackage{wrapfig}
\usepackage{amsthm}
\usepackage{amsmath}
\usepackage{multirow}
\usepackage{microtype}
\newtheorem{theorem}{Theorem}

\newtheorem{lemma}[theorem]{Lemma}
\newtheorem{corollary}[theorem]{Corollary}

\newcommand{\LCE}{\mathsf{LCE}}
\newcommand{\ShortLCE}{\mathsf{ShortLCE}}
\newcommand{\LongLCE}{\mathsf{LongLCE}}

\newcommand{\TST}[2]{{#2}\mathchar`-\mathsf{TST}({#1})}
\newcommand{\Str}{\mathit{str}}
\newcommand{\Substr}{\mathit{Substr}}

\newcommand{\LZ}{\mathsf{LZ}}
\newcommand{\locc}{\mathit{locc}}
\newcommand{\code}{\mathit{code}}
\newcommand{\polylog}{\mathrm{polylog}}

\title{
  Small-space encoding LCE data structure\\ with constant-time queries
}
\date{}

\author{Yuka Tanimura\quad
        Takaaki Nishimoto\quad
        Hideo Bannai\quad\\
        Shunsuke~Inenaga\quad
        Masayuki Takeda\\
        {Department of Informatics, Kyushu University, Japan}\\
        {\texttt{\{yuka.tanimura,takaaki.nishimoto,bannai,inenaga,takeda\}@inf.kyushu-u.ac.jp}}
}
\begin{document}

\maketitle
\begin{abstract}
  The \emph{longest common extension} (\emph{LCE}) problem is
  to preprocess a given string $w$ of length $n$ so that
  the length of the longest common prefix between suffixes of $w$ that
  start at any two given positions is answered quickly.
  In this paper, we present a data structure of
  $O(z \tau^2 + \frac{n}{\tau})$ words of space
  which answers LCE queries in $O(1)$ time
  and can be built in $O(n \log \sigma)$ time,
  where $1 \leq \tau \leq \sqrt{n}$ is a parameter,
  $z$ is the size of the Lempel-Ziv 77 factorization of $w$
  and $\sigma$ is the alphabet size.
  This is an \emph{encoding} data structure,
  i.e., it does not access the input string $w$ when answering queries
  and thus $w$ can be deleted after preprocessing.
  On top of this main result,
  we obtain further results using (variants of) our LCE data structure,
  which include the following:
  \begin{itemize}
  \item
    For highly repetitive strings where the $z\tau^2$ term
    is dominated by $\frac{n}{\tau}$,
    we obtain a \emph{constant-time and sub-linear space}
    LCE query data structure.
 
  \item
    Even when the input string is not well compressible
    via Lempel-Ziv 77 factorization, 
    we still can obtain a \emph{constant-time and sub-linear space}
    LCE data structure for suitable $\tau$ and for $\sigma \leq 2^{o(\log n)}$.
  
  \item
    The time-space trade-off lower bounds for the LCE problem
    by Bille et al.~[J. Discrete Algorithms, 25:42-50, 2014]
    and by Kosolobov~[CoRR, abs/1611.02891, 2016]
    can be ``surpassed'' in some cases
    with our LCE data structure.
    
  \end{itemize}
\end{abstract}

\section{Introduction}

\subsection{The LCE problem}
The \emph{longest common extension} (\emph{LCE}) problem is
to preprocess a given string $w$ of length $n$ so that
the length of the longest common prefix of
suffixes of $w$ starting at two query positions is answered quickly.
The LCE problem often appears as a sub-problem of
many different string processing problems,
e.g.,
approximate pattern matching~\cite{LandauV86,GalilG86},
string comparison~\cite{LandauMS98},
and finding string regularities such as
maximal repetitions (a.k.a. runs)~\cite{KolpakovK99,BannaiIINTT15},
distinct squares~\cite{GusfieldS04,BannaiIK16},
gapped repeats~\cite{BrodalLPS99,KolpakovK00,GawrychowskiIIK16,CrochemoreKK16},
palindromes and gapped palindromes~\cite{Gusfield97,KolpakovK09,NarisadaDNIS17},
and 2D palindromes~\cite{GeizhalsS16}.

A well known solution to the LCE problem is achieved by the suffix tree~\cite{Weiner73} augmented with
a constant-time linear-space longest common ancestor (LCA) data structure~\cite{bender00:_lca_probl_revis},
or equivalently the inverse suffix array (ISA)
and longest common prefix (LCP) array
augmented with a constant-time range minimum query (RMQ) data structure~\cite{manber93:_suffix_array,bender00:_lca_probl_revis}.
Either combination uses $O(n)$ words of space,
answer LCE queries in $O(1)$ time,
and can be constructed using $O(n)$ words of working space,
in $O(n)$ time for integer alphabets or in $O(n \log \sigma)$ time
for general ordered alphabets of size $\sigma$.
The $O(n)$ space requirements, however, can be prohibitive for
massive text, and hence the main focus of recent research has been
on more space-efficient solutions with trade-offs for query time.

\subsection{Space-efficient LCE data structures: Indexing or encoding}
In this paper, we will call data structures that use $o(n)$ words, or equivalently $o(n\log n)$ bits as \emph{sub-linear space} data structures.
Bille et al.~\cite{bille14:_time}
proposed the first \emph{sub-linear} space LCE query data structure
which occupies $O(\frac{n}{\tau})$ words of space,
answers LCE queries in $O(\tau^2)$ time,
and can be built in $O(\frac{n^2}{\tau})$ time using
$O(\frac{n}{\tau})$ words of working space,
for parameter range $1 \leq \tau \leq \sqrt{n}$.
Bille et al.~\cite{bille15:_longes_common_exten_sublin_space}
developed an improved sub-linear space data structure
which occupies $O(\frac{n}{\tau})$ words of space,
answers LCE queries in $O(\tau)$ time,
and can be built in $O(n^{\frac{3}{2}})$ \emph{expected} time
using $O(\frac{n}{\tau})$ words of working space,
or in $O(n^{2+\epsilon})$ time using $O(\frac{n}{\tau})$ words of working space
for parameter $1 \leq \tau \leq n$,
where $0 < \epsilon < 1$.
Tanimura et al.~\cite{TanimuraIBIPT16} proposed
an LCE data structure of $O(\frac{n}{\tau})$ words of space,
which can be built in faster $O(n \tau)$ time
using $O(\frac{n}{\tau})$ words of working space,
but takes slower $O(\tau \log \min\{\tau,\frac{n}{\tau}\})$ time for LCE queries,
for parameter $1 \leq \tau \leq n$.
All of these sub-linear space LCE data structures are \emph{indexing} data structures~\cite{BrodalDR12}, that is,
access to the input string is required to answer queries.
Therefore, these data structures require extra
$n\lceil \log \sigma \rceil$ \emph{bits} of space for
storing the input string.
A space-efficient indexing LCE data structure based on fingerprints
is also proposed~\cite{Prezza16}.

There also exist \emph{compressed} LCE data structures
which store a compressed form of the input string
represented as a straight-line program (a.k.a. grammar-based text compression)~\cite{DBLP:conf/mfcs/NishimotoIIBT16,Inenaga15,I16}.
Unlike the afore-mentioned indexing LCE data structures,
these methods do not need to keep the original uncompressed input string.
In this sense, they can be seen as \emph{encoding} data structures~\cite{BrodalDR12} for the LCE problem.
For compressible strings, the space usage of these data structures can be sub-linear.

\begin{table}[t]
    \caption{Deterministic LCE query data structures. $n$ is the length of the input string, $\sigma$ is the alphabet size, $z$ is the size of the Lempel-Ziv 77 factorization of $w$, $l$ is the length of the LCE, $\omega$ is the machine word size, $\epsilon > 0$ is an arbitrarily small constant, and $\tau$ is a trade-off parameter ($\dagger : 1 \leq \tau \leq n$,
	 $\diamond : 1 \leq \tau \leq \sqrt{n}$). ISA+ consists of the inverse suffix array of $w$, the LCP array and the RMQ data structure. $\star$ is valid for $\omega = \Theta(\log n)$ and $\sigma \leq 2^{o(\log n)}$.}
	\label{table:resultcomparison}    
	\tabcolsep = 1mm
           \centerline{
	\begin{tabular}{|l|l|l|l|p{1.23cm}|}\hline
		\multicolumn{2}{|c|}{Data structure} & \multicolumn{2}{|c|}{Preprocessing} & \multirow{2}{*}{Ref} \\\cline{1-4}
		Space (bits) & Query Time & Working space & Construction time &\\\hline
		$O(\omega)$ & $O(n)$ & $O(\omega)$ & -  & na\"ive\\ \hline
		$O(n\omega)$ & $O(1)$ & $O(n\omega)$ & $O(n)$  & ISA+\\ \hline
		$n \lceil \log \sigma \rceil + O(\frac{n\omega}{\tau})$ & $O(\tau^2)$ & $O(\frac{n\omega}{\tau})$ & $O(n^2/\tau)$ & $\diamond$ \cite{bille14:_time}\\\hline
		$n \lceil \log \sigma \rceil + O(\frac{n\omega}{\tau})$ & $O(\tau)$ & $O(\frac{n\omega}{\tau})$ & $O(n^{3/2})$ exp. & 
		$\dagger$ \cite{bille15:_longes_common_exten_sublin_space}(1)\\\hline
		$n \lceil \log \sigma \rceil + O(\frac{n\omega}{\tau})$ & $O(\tau)$ & $O(\frac{n\omega}{\tau})$ & $O(n^{2+\epsilon})$ & 
		$\dagger$ \cite{bille15:_longes_common_exten_sublin_space}(2)\\
		\hline
		$n\lceil \log \sigma \rceil + O(\frac{n\omega}{\tau})$ & $O(\tau\log\min\{\tau,\frac{n}{\tau}\})$ & $O(\frac{n\omega}{\tau})$ & $O(n\tau)$  & $\dagger$ \cite{TanimuraIBIPT16}\\
		\hline		
		$n\lceil \log \sigma \rceil + O(\omega\log n)$ & $O(\log l)$ & $O(\omega\log n)$ & $O(n\log n)$ exp.  & \cite{Prezza16}\\
		\hline
		$O(z \omega\log n \log^* n)$ & $O(\log n \log^* n)$ & $O(z \omega\log n \log^* n)$ & $O(n \log \sigma)$  & \cite{DBLP:conf/mfcs/NishimotoIIBT16},\cite{Inenaga15}\\
		\hline
		$O(z \omega\log \frac{n}{z})$ & $O(\log n)$ & $O(n\omega)$ & $O(n)$  & \cite{I16}\\
		\hline
		$O(z \omega\log \frac{n}{z})$ & $O(\log n)$ & $O(n\log\sigma + z\omega \log \frac{n}{z})$ & $O(n\log\log\sigma + z\log^2\frac{n}{z})$ & \cite{I16}+\cite{KopplS15} \\
		\hline
		$O((z \tau^2 + \frac{n}{\tau})\omega)$ & $O(1)$ & $O((z \tau^2 + \frac{n}{\tau})\omega)$ & $O(n \log \sigma)$  & $\diamond$ ours\\
                \hline
		$O(z^{1/3}n^{2/3}\omega)$ & $O(1)$ & $O(z^{1/3}n^{2/3}\omega)$ & $O(n \log \sigma\log n)$  & ours\\
                \hline
                      $o(n \log n)$ & $O(1)$ & $o(n \log n)$ & $o(n \log^2 n)$ & $\star$ ours\\
		\hline
		$O(\sqrt{nz}\omega)$ & $O(\sqrt{\frac{n}{z}})$ & 
		$O(\sqrt{nz} \omega)$ & $O(n \log \sigma \log n)$  & ours\\
		\hline
	\end{tabular}
           }	
\end{table}

\subsection{Our LCE data structure: Constant-time queries, sub-linear space, and encoding}

This paper proposes the \emph{first $O(1)$-time} LCE data structure
which takes \emph{sub-linear space} in several reasonable cases, namely, when the
string is compressible, and/or, when the alphabet size is suitably small.
Our data structure has both flavours of
sub-linear space and compressed LCE data structures.
Namely, for parameter $1 \leq \tau \leq \sqrt{n}$,
we present an LCE data structure which takes $O(z\tau^2 + \frac{n}{\tau})$
words of space,
answers LCE queries in $O(1)$ time,
and can be built in $O(n \log \sigma)$ time for
general ordered alphabets of size $\sigma$
using $O(z\tau^2 + \frac{n}{\tau})$ words of working space,
where $z$ is the size of the Lempel-Ziv~77 factorization~\cite{LZ77}
of the input string.
It is known that $z$ is a lower bound of the size of \emph{any} grammar-based compression of the string~\cite{Rytter03}, 
and can be very small for highly repetitive strings. In such cases where the $z\tau^2$ term is dominated by $\frac{n}{\tau}$, 
our LCE data structure uses sub-linear space.
An interesting feature is that we do \emph{not} actually compress the input string,
i.e., do not compute the Lempel-Ziv~77 factorization,
but we construct a data structure
whose size is bounded by $O(z\tau^2 + \frac{n}{\tau})$.

Even when the input string is not well compressible via Lempel-Ziv 77,
for suitably small alphabets, we can build a sub-linear space LCE data structure 
with $O(1)$ query time using appropriate values of $\tau$.
By choosing $\tau = (\frac{n}{z})^{\frac{1}{3}}$,
our LCE query data structure takes
$O(z^{\frac{1}{3}} n^{\frac{2}{3}})$ words of space,
which translates to
$O(n / (\log_\sigma n)^{\frac{1}{3}})$
using the well-known fact that $z = O(n / \log_\sigma n)$.
This means that our data structure can be stored in
$O(n \log n / (\log_\sigma n)^{\frac{1}{3}}) = O(n (\log n)^{\frac{2}{3}}(\log \sigma)^{\frac{1}{3}})$ \emph{bits} of space.
This implies that for alphabets of size $\sigma \leq 2^{o(\log n)}$
(note that these contain polylogarithmic alphabets),
our data structure takes only $o(n \log n)$ bits of space,
yet answers LCE queries in $O(1)$ time.
Also, our LCE data structure does not access the input string
when answering queries,
and hence the input string does not have to be kept.
To our knowledge, this is the first sub-linear space encoding
LCE data structure for strings incompressible with Lempel-Ziv 77.

The key to our efficient LCE query data structure is a hybrid use
of the \emph{truncated suffix trees}~\cite{NaAIP03:TST}
and block-wise LCE queries based on $t$-covers~\cite{PuglisiT08,bille14:_time}.
The $q$-truncated suffix tree of a string $w$
is the compact trie (a.k.a. Patricia tree) which represents
all substrings of $w$ of length at most $q$.
We observe that,
for any $1 \leq q \leq n$,
the $q$-truncated suffix tree can be stored in $O(zq)$ words of space,
including a string to which the edges label pointers refer.
We also show that the block-wise LCE query data structure
based on $t$-covers can be efficiently built by
the $t$-truncated suffix tree, leading to our result.
Several variants of our data structure are considered, as summarized in Table~\ref{table:resultcomparison}.

The rest of this paper is organized as follows.
Section~\ref{sec:preliminaries} gives some definitions  and
introduces tools which will be used as building-blocks of our LCE data structure.
In Section~\ref{sec:LCE_data_structure}
we propose our new LCE data structure and analyze its time/space complexities.
In Section~\ref{sec:lowerbound} we review some lower bounds
on the LCE problem and show that using our LCE data structure,
these lower bounds can be ``surpassed'' in some cases.
We conclude in Section~\ref{sec:conclusions}.

\section{Preliminaries} \label{sec:preliminaries}

\subsection{Notations}

Let $\Sigma$ be an ordered alphabet of size $\sigma$.
Each element of $\Sigma^*$ is called a \emph{string}.
The length of a string $w$ is denoted by $|w|$.
The empty string $\varepsilon$ is the string of length zero,
namely $|\varepsilon| = 0$.
If $w = xyz$ for some strings $w,x,y,z$, then $x$, $y$, and $z$ are respectively
called a {\em prefix}, {\em substring}, and {\em suffix} of $w$.
For any $1 \leq i \leq |w|$, let $w[i]$ denote the $i$th character of $w$.
For any $1 \leq i \leq j \leq |w|$, let $w[i..j]$ denote
the substring of $w$ that begins at position $i$ and ends at position $j$,
namely, $w[i..j] = w[i] \cdots w[j]$.
A string of length $q$ is called a \emph{$q$-gram}.
For any $1 \leq q \leq |w|$,
let $\Substr_q(w)$ denote the set of all $q$-grams occurring in $w$
and the $q-1$ suffixes of $w$ of length shorter than $q$,
namely, $\Substr_q(w) = \{w[i..\min\{i+q-1, |w|\}] \mid 1 \leq i \leq |w|\}$.

For any string $w$, let $\LCE^w(i,j)$ denote the length of the longest common
prefix of $w[i..|w|]$ and $w[j..|w|]$.
We will write $\LCE(i,j)$ when $w$ is clear from the context.
Since $\LCE^w(i,i) = |w|-i$, we will only consider the case when $i\neq j$.
For any integers $i\leq j$, let $[i..j]$ denote the set of integers
from $i$ to $j$ (including $i$ and $j$).

The \emph{Lempel-Ziv 77 factorization with self-references}~\cite{LZ77} 
of a string $w$ is a sequence $\LZ(w) = f_1, \ldots, f_z$ of
$z$ non-empty substrings of $w$ such that $w = f_1 \cdots f_z$
and for $1 \leq i \leq z$,
\begin{itemize}
  \item $f_{i} = w[|f_1 \cdots f_{i-1}| + 1] \in \Sigma$ if $s[|f_1 \cdots f_{i-1}| + 1]$ is a character not occurring in $f_0 \cdots f_{i-1}$,
  \item $f_{i}$ is the longest prefix of $f_{i} \cdots f_{z}$ such that $f_{i}$ is a substring of $w$ beginning at a position in range $[1..|f_1 \cdots f_{i-1}|]$,
\end{itemize}
where $f_{0} = \varepsilon$.
The \emph{size} of $\LZ(w)$ is the number $z$ of factors $f_1, \ldots, f_z$,
and is denoted as $|\LZ(w)| = z$.
For instance, for string $w = \mathtt{abababcabababcabababcd}$ of length $22$,
$\LZ(w) = \mathtt{a,b,abab,c,abababcabababc,d}$ and $|\LZ(w)| = 6$.

Our model of computation is a standard word RAM with machine word size $\omega \geq \log n$.
The space requirements will be evaluated by the number of words unless otherwise stated.

\subsection{Tools}

We will use the following tools
as building blocks of our LCE data structure.

\vspace*{1pc}
\noindent \textbf{$t$-covers.}
For any positive integer $t$,
a set $D \subseteq [0..t-1]$ is called a \emph{$t$-difference-cover}
if $[0..t-1] = \{(x-y) \bmod t \mid x, y \in D\}$,
namely, every element in $[0..t-1]$ can be expressed by
a difference between two elements in $D$ modulo $t$.
For any positive integer $n$,
a set $S \subseteq [1..n]$ is called a \emph{$t$-cover} of $[1..n]$
if $S = \{i \in [1..n] \mid (i \bmod t) \in D\}$
with some $t$-difference-cover $D$,
and there is a constant-time computable 
function $h(i, j)$ that for any $1 \le i, j \le n-t$,
$0 \le h(i, j) \le t$ and $i + h(i, j), j + h(i, j) \in S$.

\begin{lemma}[\cite{Maekawa85}]
  For any integer $t$, there exists a $t$-difference-cover
  of size $O(\sqrt{t})$ which $D(t)$ can be computed in $O(\sqrt{t})$ time.
\end{lemma}

\begin{lemma}[\cite{BurkhardtK03}] \label{lem:t-cover}
  For any integer $t~(\le n)$, there exists a $t$-cover of size
  $O(\frac{n}{\sqrt{t}})$ which can be computed in $O(\frac{n}{\sqrt{t}})$ time.
\end{lemma}

In what follows, we will denote by $S(t)$
an arbitrary $t$-cover of $[1..n]$ which satisfies
the conditions of Lemma~\ref{lem:t-cover}.
See Figure~\ref{fig:covers} for an example of a $t$-cover $S(t)$.
\begin{figure}[t]
\centerline{
  \includegraphics[width=0.7\linewidth]{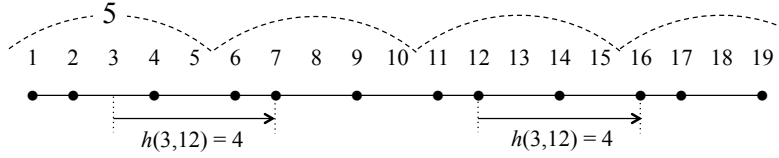}
}
\caption{Let $t=5$ and $D=\{1,2,4\}$.
  This figure shows an example of a 5-cover $S(5) = \{1,2,4,6,7,9,11,12,14,16,17,19\}$.
  The black dots represent the elements in $S(5)$.
  For instance, we have $h(3, 12) = 4$, namely, $3+4, 12+4 \in S(5)$.}
\label{fig:covers}
\end{figure}

\vspace*{1pc}
\noindent \textbf{Truncated suffix trees.}
For convenience, we assume that 
any string $w$ ends with a special end-marker $\$$
that appears nowhere else in $w$.
Let $n = |w|$.
For any $1 \leq q \leq n$,
the \emph{$q$-truncated suffix tree} of $w$, 
denoted $\TST{w}{q}$, is a Patricia
tree which represents $\Substr_q(w)$.
Namely,
$\TST{w}{q}$ is an edge-labeled rooted tree such that:
(1) Each edge is labeled with a non-empty substring of $w$;
(2) Each internal node $v$ has at least two children,
  and the labels of the out edges of $v$ begin with distinct characters;
(3) For any leaf $u$, there is at least one position $1 \leq i \leq  n$
  such that $w[i.. \min\{i+q-1, n\}]$ is the string obtained by
  concatenating the edge labels from the root to $u$;
(4) For any position $1 \le i \le n$ in $w$,
  there is a unique leaf $u$ such that
  $w[i..\min\{i+q-1, n\}]$ is the string obtained by
  concatenating the edge labels from the root to $u$.
\begin{wrapfigure}[11]{r}{55mm}
  \centerline{
    \includegraphics[width=50mm]{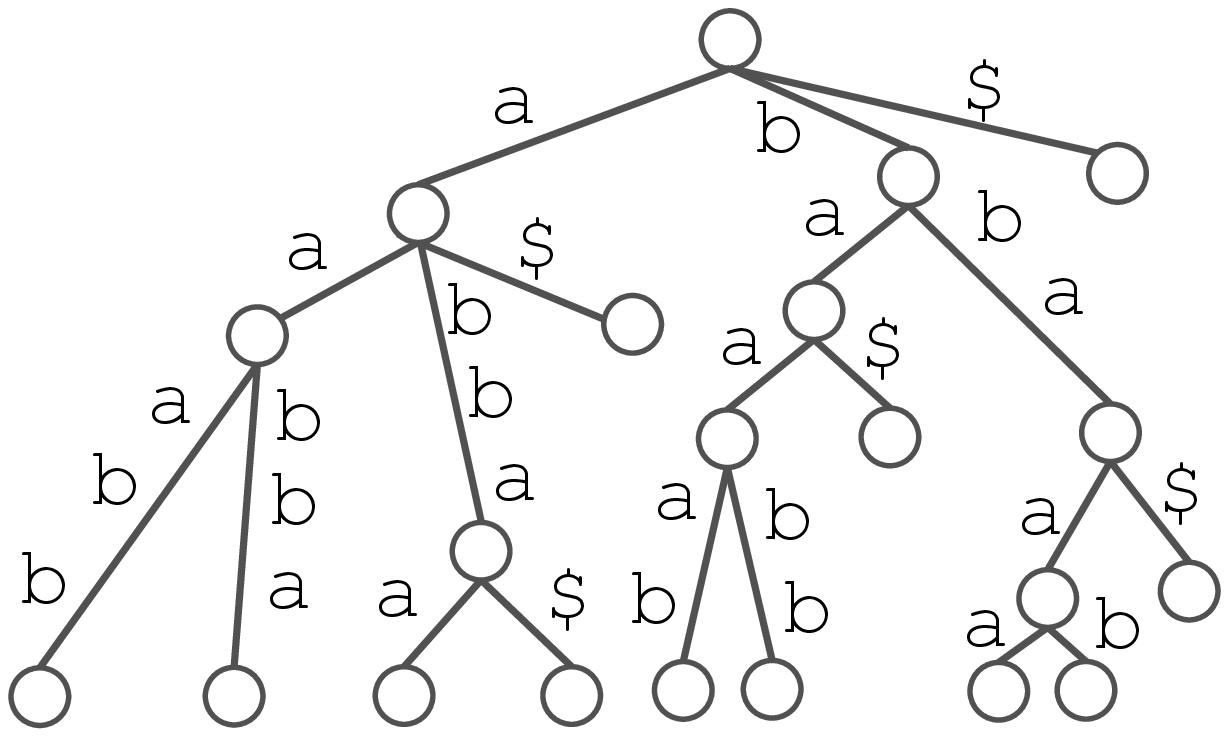}
  }
 \caption{$\TST{w}{5}$
  with string $w = \mathtt{baabbaabbaaabbaabba\$}$.}
 \label{fig:t-TST}
\end{wrapfigure}

Informally speaking,
$\TST{w}{q}$ can be obtained by trimming the full suffix tree of $w$
so that any path from the root represents a substring of at most $q$.
Clearly, the number of leaves in $\TST{w}{q}$ is equal to $|\Substr_q(w)|$.
We assume that the leaves of $\TST{w}{q}$ are sorted in lexicographical order.
Figure~\ref{fig:t-TST} shows an example of a $\TST{w}{q}$.
For any node $u$ of $\TST{w}{q}$,
$\Str(u)$ denotes the string spelled out by the path from the root to $u$.

In the case of the full suffix tree ($\TST{w}{n}$) of string $w$ of length $n$,
each edge label $x$ is represented by a pair $(i,j)$ 
of positions in $w$ such that $x = w[i..j]$.
We call $w$ as the \emph{reference string} for the full suffix tree,
and this way the full suffix tree can be stored in $O(n)$ space.
For $\TST{w}{q}$, Vitale et al.~\cite{VitaleMS15} showed how to represent $\TST{w}{q}$
in $O(|\Substr_q(w)|)$ space, including the reference string,
and how to construct them efficiently, both in time and space.
\begin{lemma}[\cite{VitaleMS15}] \label{lem:TST_const}
  Let $w$ be any string of length $n$ over an ordered alphabet of size $\sigma$.
  For any $1 \leq q \leq n$, let $y = |\Substr_q(w)|$.
  Then, the exists a reference string $w'$ of length $O(y)$ for $\TST{w}{q}$.
  Moreover, $\TST{w}{q}$ with the leaves sorted in lexicographical order,
  and a reference string $w'$
  can be constructed in $O(n\log\sigma)$ time with $O(y)$ working space.
\end{lemma}

We also show the following lemma.
\begin{lemma} \label{lem:zq_size_TST}
$\TST{w}{q}$ can be represented in $O(zq)$ space, where $z = |\LZ(w)|$.
\end{lemma}

\begin{proof}
  By Lemma~\ref{lem:TST_const},
  it suffices to show that $|\Substr_q(w)| = O(zq)$.
  For each $q$-gram $p \in \Substr_q(w)$,
  let $\locc_w(p)$ be the beginning position of the leftmost occurrence
  of $p$ in $w$.
  If $q = 1$, then clearly $|\Substr_1(w)| \leq z$
  and hence the lemma holds.
  If $q \geq 2$ then the interval $[\locc_w(p)..\locc_w(p)+q-1]$ must cross
  the boundary of two adjacent factors of $\LZ(w)$,
  since otherwise the interval is completely contained in
  a single factor of $\LZ(w)$ but
  this contradicts that $[\locc_w(p)..\locc_w(p)+q-1]$
  is the leftmost occurrence of $p$ in $w$.
  Clearly, the maximum number of $q$-grams that can 
  cross a boundary of $\LZ(w)$ is $q-1$.
  Hence, the total number of distinct $q$-grams in $w$ is $O(zq)$.
  Also, $\Substr_{q}(w)$ contains $q$ substrings $w[n-q+1], \ldots, w[n]$
  of $w$ which are shorter than $q$.
  Overall, we obtain $|\Substr_q(w)| = O(zq)$.
\end{proof}

The next theorem follows from Lemmas~\ref{lem:TST_const} and~\ref{lem:zq_size_TST} and an obvious fact that $|\Substr_q(w)| \leq n$.

\begin{theorem} \label{theo:TST_zq_const}
  Given a string $w$ of length $n$ over an ordered alphabet of size $\sigma$
  and integer $1 \leq q \leq n$,
  we can construct an $O(\min\{zq, n\})$-space representation of
  $\TST{w}{q}$ in $O(n \log \sigma)$ time with $O(\min\{zq, n\})$ working space.
\end{theorem}
In what follows, we will only consider interesting cases
where $zq < n$ for a given $1 \leq q \leq n$,
and will simply use $O(zq)$ to denote the size of $\TST{w}{q}$.
\section{Our LCE data structure} \label{sec:LCE_data_structure}
\subsection{Overview of our algorithm}

The general framework of our space-efficient LCE algorithm follows the approach
of Gawrychowski et al.'s LCE algorithm for
strings over a general ordered alphabet~\cite{GawrychowskiKRW16}.
Namely, we compute $\LCE(i,j)$ using the two following types of queries:
\begin{eqnarray*}
  \ShortLCE_t(i,j) & = & \min (\LCE(i,j),t), \\
  \LongLCE_t(i, j) & = & \begin{cases}
      \lfloor \text{LCE}(i,j)/t \rfloor & \text{if } i,j \in S(t),\\
      \bot & \text{otherwise.}
    \end{cases}
\end{eqnarray*}
$\LCE(i,j)$ is computed in the following manner.
Let $\delta = h(i, j)$.
Recall that $\delta \leq t$ can be computed in constant time
and that $i + \delta, j + \delta \in S(t)$.
First, we compare up to the first $\delta$ characters of
$w[i..|w|]$ and $w[j..|w|]$ using $\ShortLCE_t(i, j)$.
If $l_1 = \ShortLCE_t(i, j)$ is shorter than $t$, then $\LCE(i,j) = l_1$.
If $l_1 = t$, then $\LCE(i,j)$ is at least $t$ long.
To check if it further extends, 
we compute $l_2 = \LongLCE_t(i + \delta, j + \delta)$, 
and $l_3 = \ShortLCE_t(i + \delta + l_2, j + \delta + l_2)$.
Finally, we get $\LCE(i,j) = \delta + t\cdot l_2 + l_3$.
See also Figure~\ref{fig:overview}.

The main difference between Gawrychowski et al.'s method and ours is in
how to compute $\ShortLCE_t(i,j)$. 
While they use a Union-Find structure that takes $O(n)$ working space
(for $O(n)$ queries) as a main tool,
we use an augmented $\TST{w}{2t}$ for $\Substr_{2t}(w)$
which occupies $O(zt + \frac{n}{\sqrt{t}})$ total space,
answers $\ShortLCE_t(i,j)$ queries in $O(1)$ time,
and can be constructed in $O(n \log \sigma)$ time
with $O(zt)$ working space.
How to answer $\LongLCE_t(i,j)$ queries is equivalent to Gawrychowski et al.'s,
namely, we sample the positions from $S(t)$ so that
LCE queries for these sampled positions can be answered in $O(1)$ time.
We show how to build the data structure for $\LongLCE_t(i,j)$ queries
by using $\TST{w}{t}$ for $\Substr_{t}(w)$ in $O(n \log \sigma)$ time
with $O(zt)$ working space.

\begin{figure}[t]
\centerline{
  \includegraphics[scale=0.6]{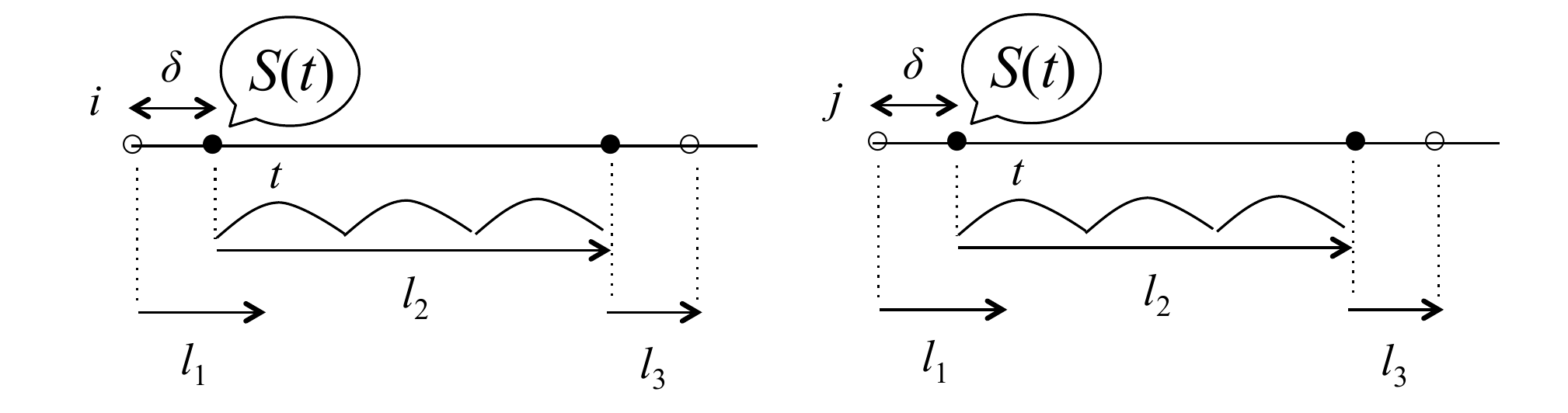}
  }
\caption{Illustration of an overview of our $\LCE(i,j)$ algorithm.
  We are given two positions $i$ and $j$ in string $w$.
  First, we compute $l_1 = \ShortLCE_t(i,j)$.
  If $l_1 < t$, then $\LCE(i,j) = l_1$. Otherwise, we compute $\LongLCE_t(i+\delta, j+\delta)$ where $i+\delta, j+\delta \in S(t)$ and $0 \leq \delta \leq t$. We finally compute $l_3 = \ShortLCE(i+\delta+l_2, j+\delta+l_2)$ where $l_2 = t\LongLCE_t(i+\delta,j+\delta)$. Then $\LCE(i,j) = \delta + l_2 + l_3$.}
\label{fig:overview}
\end{figure}

\subsection{$\boldsymbol{\ShortLCE_t}$ queries} \label{sec:short_LCE}

For $\ShortLCE_t(i, j)$ queries, we use
$\TST{w}{2t}$ which represents the set $\Substr_{2t}(w)$ of
all substrings of $w$ of length at most $2t$.
For any position $1 \leq i \leq n$,
let $p_i$ denote the substring of $w$
that begins at position $i$ and is of length at most $2t$,
namely, $p_i = w[i..\min\{i+2t-1, n\}]$.
Notice that $\Substr_{2t}(w) = \bigcup_{i=1}^{n}\{p_i\}$.
For any position $1 \leq i \leq n$ in $w$,
let $\ell(i) = u$ iff $u$ is the leaf of $\TST{w}{2t}$
such that $\Str(u) = p_i$.
Basically, we will compute $\ShortLCE_t(i, j)$ by efficiently finding
the LCA of the corresponding leaves $\ell(i)$ and $\ell(j)$ on $\TST{w}{2t}$.
The reason that we use $\TST{w}{2t}$ rather than $\TST{w}{t}$
will become clear later.

Now the key is how to access $\ell(i)$ for a given position $i$ in $w$.
As our goal is to build a sub-linear space data structure for $\ShortLCE_t$ queries,
we cannot afford to store a pointer to $\ell(i)$ from every position $1 \leq i \leq n$.
Thus, we store such a pointer only from every $t$-th positions in $w$.
We call these positions as \emph{sampled positions}.
Formally, for every sampled position
$j \in Q_{t,n} = \{1 + kt \mid 0 \leq k \leq \lceil \frac{n}{t} \rceil - 1\}$
we explicitly store a pointer from $j$ to its
corresponding leaf $\ell(j)$ on $\TST{w}{2t}$.
Also, for each position $1 \leq i \leq n$ in $w$,
let $\alpha(i) = \max\{j \in Q_{t,n} \mid j \leq i\}$.
Namely, $\alpha(i)$ is the closest sampled position in $Q_{t,n}$
to the left of $i$ (or it is $i$ itself if $i \in Q_{t,n}$).

Given a position $1 \leq i \leq n$,
$\alpha(i)$ can be computed in $O(1)$ time by a simple arithmetic.
Hence, we can access the leaf $\ell(\alpha(i))$ for the
closest sampled position $\alpha(i)$ in $O(1)$ time.
The next task is to locate $\ell(i)$.
To describe our constant-time algorithm,
let us consider a conceptual DAG $G = (V, E)$ such that
$V = \Substr_{2t}(w)$ and 
$E = \{(u, c, v) \mid u[1..2t-1] = v[2..2t], c = v[1] \}$, 
where $(u, c, v)$ represents a directed edge labeled $c$ from $u$ to $v$.
This DAG $G$ is equivalent to the edge-reversed de Bruijn graph of
order $2t$, with extra nodes for the $2t-1$ suffixes of $w$
which are shorter than $2t$.
It is clear that there is a one-to-one correspondence between
the leaves of $\TST{w}{2t}$ and the nodes of the DAG $G$.
Thus, we will identify each leaf of $\TST{w}{2t}$ with the nodes of DAG $G$.

\begin{lemma} \label{lem:de_Bruijn_graph_const}
  Given $\TST{w}{2t}$ for a string $w$ of length $n$
  and for any $1 \leq 2t \leq n$,
  we can construct the DAG $G$ in $O(n)$ time using $O(zt)$ working space.
\end{lemma}

\begin{proof}
  The de Bruijn graph of order $q$ for a string of length $n$
  can be constructed in $O(n)$ time using space linear in the size
  of the output de Bruijn graph,
  provided that $\TST{w}{q}$ is already constructed~\cite{CazauxLR15}.
  By setting $q = 2t$, adding extra $2t-1$ nodes
  for the suffixes that are shorter than $2t$,
  and reversing all the edges,
  we obtain our DAG $G = (V, E)$.

  The number of nodes in $V$ is clearly equal to $|\Substr_{2t}(w)|$.
  Also, since each edge in $E$ corresponds to a distinct substring in
  $\Substr_{2t+1}(w)$,
  the number of edges in $E$ is equal to $|\Substr_{2t+1}(w)|$.
  By a similar argument to the proof of Lemma~\ref{lem:zq_size_TST},
  we obtain $|V| = |\Substr_{2t}(w)| = O(zt)$ and
  $|E| = |\Substr_{2t+1}(w)| = O(zt)$.
\end{proof}

Let $d = i-\alpha(i)$.
A key observation here is that there is a path of length $d$ from node
$p_i$ to node $p_{\alpha(i)}$ in this DAG $G$.
Since $G$ is a DAG, however, it is not easy to quickly
move from $p_{\alpha(i)}$ to $p_{i}$.
To overcome this difficulty, we consider a spanning tree
of $G$ of which the root is $p_{n} = w[n] = \$$.
Let $T$ denote any spanning tree of $G$.
See Figure~\ref{fig:doburin}
for examples of the DAG $G$ and its spanning tree $T$.
Although some edges are lost in spanning tree $T$,
it is enough for our purpose.
Namely, the following lemma holds.
\begin{lemma} \label{lem:correctness_de_Bruijn_tree}
Any spanning tree $T$ of $G$ satisfies the following properties:
(1) There is a non-branching path of length $2t$ from the root $p_n$
    to the node $p_{n-2t}$.
(2) For any $1 \leq i \leq n-2t-1$ and $0 \leq d < t$,
    let $g$ be the $d$-th ancestor of $p_{\alpha(i)}$.
    Then, $g[1..t] = p_i[1..t]$.
\end{lemma}

\begin{proof}
  The first property is immediate from the fact that
  the last character $w[n] = \$$ occurs nowhere else in $w$
  and the root represents $p_n = \$$.

  Since $d < t$ and $|p_{\alpha(i)}| = 2t$,
  we have $p_{\alpha(i)}[d..d+t-1] = p_{i}[1..t]$.
  By the first property and $\alpha(i) \leq i \leq n-2t-1$,
  the depth of node $p_{\alpha(i)}$ is at least $2t$.
  Also, by following the in-coming edge of each node 
  in the reversed direction,
  we delete the first character of the corresponding string.
  Hence, $p_{i}[1..t]$ is a prefix of 
  the $d$-th ancestor $g$ of $p_{\alpha(i)}$.
\end{proof}

We are ready to show the main result of this section.
\begin{theorem} \label{theo:short_LCE}
For any string $w$ of length $n$ and integer $1 \le t \le n$, 
a data structure of size $O(zt + \frac{n}{t})$ can be constructed 
in $O(n \log \sigma)$ time using $O(zt)$ working space
such that subsequent $\ShortLCE_t(i,j)$ queries for any $1 \le i,j \le n$ can be answered in $O(1)$ time, 
where $z = |\LZ(w)|$.
\end{theorem}

\begin{proof}
We use a spanning tree $T$ enhanced with
a \emph{level ancestor} data structure~\cite{BenderF04}
which can be constructed in time and space linear
in the size of the input tree $T$.

Given two positions $i,j$ in $w$,
we answer $\ShortLCE_{t}(i,j)$ query as follows:
\begin{enumerate}
\item Compute the closest sampled positions $\alpha(i)$ and $\alpha(j)$ by simple arithmetics.
\item Access the nodes $p_{\alpha(i)}$ and $p_{\alpha(j)}$ in the spanning tree $T$ using pointers from the sampled positions $\alpha(i)$ and $\alpha(j)$,
  respectively.
\item Let $d = i-\alpha(i)$ and $d' = j-\alpha(j)$.
  Access the $d$-th ancestor $u$ of $p_{\alpha(i)}$ and
  the $d'$-th ancestor of $p_{\alpha(j)}$ using level ancestor queries on $T$.
\item Compute the LCA $x$ of the two leaves $u$ and $v$
  on $\TST{w}{2t}$, and return $\min\{|\Str(x)|, t\}$.
\end{enumerate}

The correctness follows from Lemma~\ref{lem:correctness_de_Bruijn_tree}.
Since each step of the above algorithm takes $O(1)$ time,
we can answer $\ShortLCE_t(i, j)$ in $O(1)$ time.
By Lemma~\ref{lem:zq_size_TST},
the size of $\TST{w}{2t}$ with an LCA data structure is $O(zt)$,
and also the size of the spanning tree $T$
with a level ancestor data structure is $O(|\Substr_{2t}(w)|) = O(zt)$.
In addition, we store pointers from the $\Theta(\frac{n}{t})$ sampled
positions to their corresponding nodes in $T$.
Overall, the total space requirement of our data structures
is $O(zt + \frac{n}{t})$.
We can build these data structures in a total of $O(n \log \sigma)$ time
using $O(zt)$ working space
by Theorem~\ref{theo:TST_zq_const} and Lemma~\ref{lem:de_Bruijn_graph_const}.
\end{proof}

\begin{figure}[tb]
	\centerline{
		\includegraphics[width=0.45\linewidth]{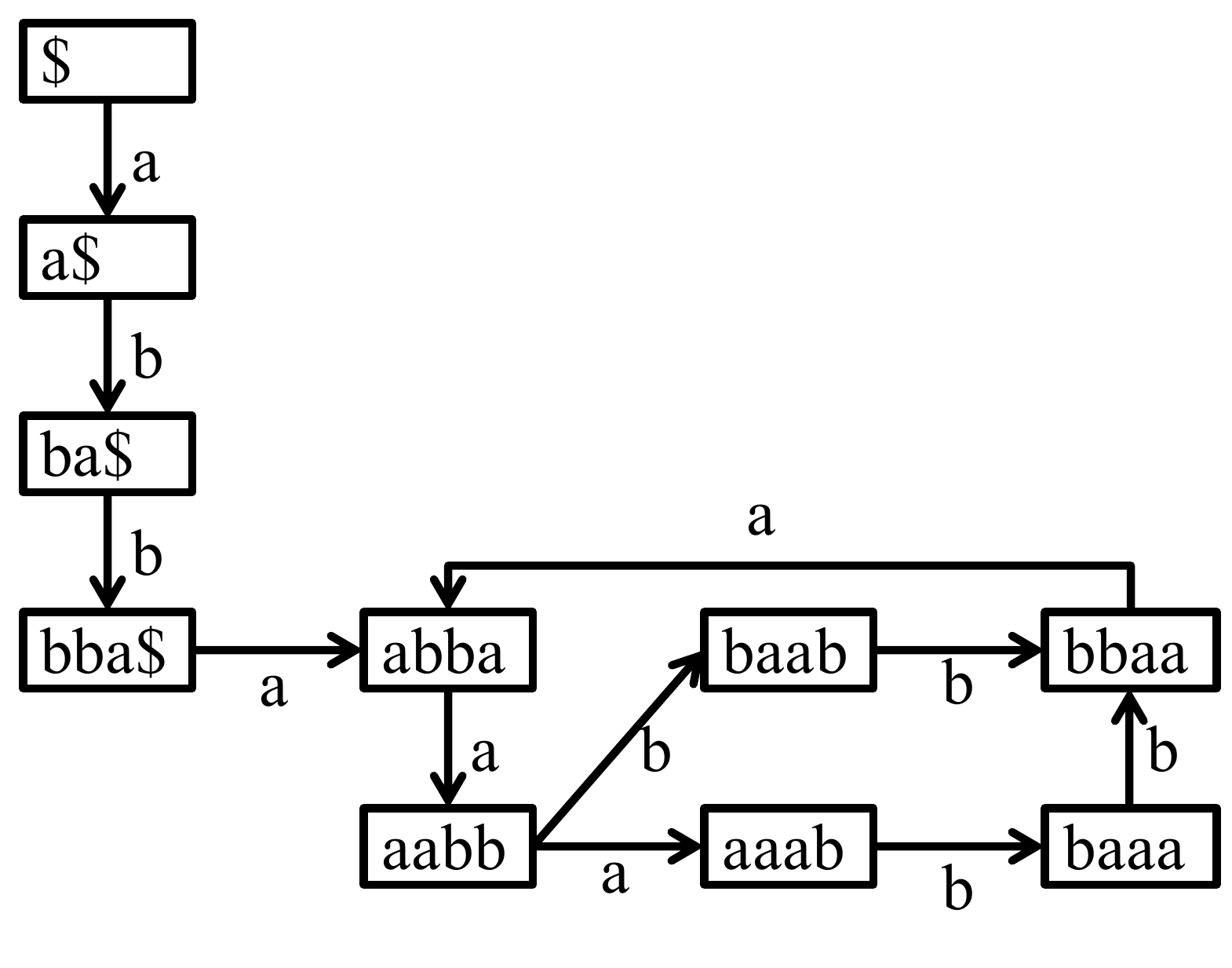}
		\includegraphics[width=0.45\linewidth]{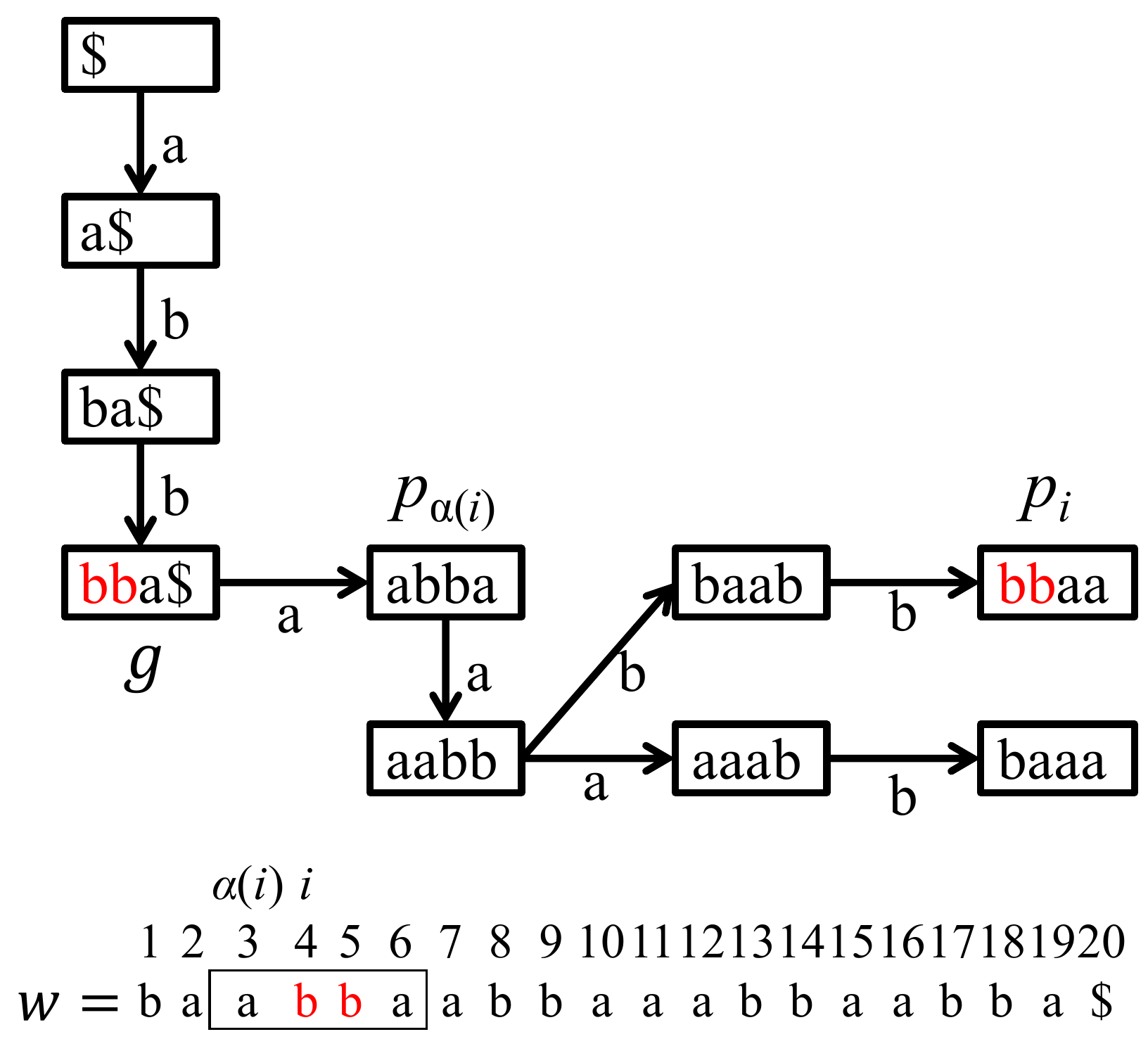}
                }
		\caption{The left graph $G$ is the edge-reversed de Bruijn graph of order $2t$, 
			with extra nodes for the $2t-1$ suffixes of $w$ which are shorter than $2t$, 
			where $t = 2$ and $w$ is the same string as in Figure~\ref{fig:t-TST}. 
			An edge from $u$ to $v$ labeled character $c$ represents $c \cdot u[1..2t-1] = v$.
			The right tree is a spanning tree of the left graph. 
			Let $i = 4$ and $\alpha(i) = 3$. 
			Then $p_{i} = bbaa$, $p_{\alpha(i)} = abba$. 
			Let $g$ be the $d$-th ancestor of $p_{\alpha(i)}$ in the right tree, where $0 \leq d < t$. 
			Then $g[1..t] = p_{i}[1..t]$ holds by Lemma~\ref{lem:correctness_de_Bruijn_tree}.}
		\label{fig:doburin}
\end{figure}

\subsection{$\boldsymbol{\LongLCE_t}$ queries}
\label{sec:long}

At a high level, 
our $\LongLCE_t(i, j)$ query algorithm is an adaptation of
the $t$-cover based algorithm by Puglisi and Turpin~\cite{PuglisiT08},
which was later re-discovered by Bille et al.~\cite{bille14:_time}.
Gawrychowski et al.~\cite{GawrychowskiKRW16} showed
that an $O(\frac{n}{\sqrt{t}})$-space data structure, which answers
$\LongLCE_t(i, j)$ query in $O(1)$ time, can be constructed in
$O(n \log t)$ time with $t = \Omega(\log^2 n)$
for a string of length $n$ over a general ordered alphabet.
In this section,
we show the same data structure as Gawrychowski et al.
can be constructed in $O(n \log \sigma)$ time with $O(zt + \frac{n}{t})$
working space for a general ordered alphabet of size $\sigma$
and any $1 \leq t \leq n$.

Consider a $t$-cover $S(t)$ of $[1..n]$ for some $t$-difference-cover $D$.
For each position $i \in S(t)$ such that $i+t-1 \leq n$,
the substring $b_i = w[i..i+t-1]$ is said to be a \emph{$t$-block}.
The goal here is to answer the block-wise LCE value
$\LongLCE_t(i,j)$ for two given positions in the $t$-cover $S(t)$.
Since we query $\LongLCE_t(i, j)$ only for positions $i, j \in S(t)$
and the answer to $\LongLCE_t(i, j)$ is a multiple of $t$,
we can regard each $t$-block as a single character.
Thus, we sort all $t$-blocks in lexicographical order,
and encode each $t$-block by its lexicographical rank.
Since each $t$-block is of length $t$,
we can sort the $t$-blocks in $O(\frac{n}{\sqrt{t}} \log \frac{n}{\sqrt{t}})$ time with $O(zt + \frac{n}{\sqrt{t}})$ working space
by using any suitable comparison-based sorting algorithm and
our $O(1)$-time $\ShortLCE_t$ query data structure of Section~\ref{sec:short_LCE}.
The next lemma shows that we can actually compute the lexicographical ranks
of all $t$-blocks more efficiently.

\begin{lemma} \label{lem:sort_t-blocks_1}
  Let $w$ be an input string of length $n$ and $1 \leq t \leq n$ be an integer. 
  Given the data structure for $\ShortLCE_t$ queries of
  Theorem~\ref{theo:short_LCE} for $w$,
  we can sort all $t$-blocks of $w$ in lexicographic order
  in $O(zt + \frac{n}{\sqrt{t}})$ time 
  using $O(zt + \frac{n}{t})$ working space, where $z = |\LZ(w)|$.
\end{lemma}

\begin{proof}
  We insert new (non-branching) nodes to $\TST{w}{2t}$ such that
  every $t$-gram in $w$ is represented by an explicit node.
  This increases the size of the tree by a constant factor.
  We also associate each node $u$ such that $|\Str(u)| = t$
  with the lexicographical rank of the $t$-gram $\Str(u)$
  among all $t$-grams in $w$.
  Then, we associate each leaf $\ell$ of the tree such that $|\Str(\ell)| \geq t$
  with its ancestor $v$ which represents a $t$-gram.
  All these can be preformed in $O(zt)$ total time
  by standard depth-first traversals on the tree.

  Then, for each $t$-block $b_i = w[i..i+t-1]$,
  we can access a leaf $\ell$ of $\TST{w}{2t}$ such that
  $\Str(\ell)[1..t] = w[i..i+t-1]$ in $O(1)$ time using the algorithm
  of Theorem~\ref{theo:short_LCE},
  and we return the rank of the ancestor $v$ of $\ell$
  that represents $b_i = w[i..i+t-1]$.
  Since there are $O(\frac{n}{\sqrt{t}})$ $t$-blocks in $w$,
  it takes a total of $O(zt + \frac{n}{\sqrt{t}})$ time.
  The working space is $O(zt + \frac{n}{t})$ by Theorem~\ref{theo:short_LCE}.
\end{proof}

There is an alternative algorithm
to sort the $t$-blocks, as follows:
\begin{lemma} \label{lem:sort_t-blocks_2}
For any string $w$ of length $n$ over an alphabet of size $\sigma$,
any integer $1 \leq t \leq n$, 
we can sort all $t$-blocks in lexicographic order in $O(n\log\sigma)$ time 
using $O(zt)$ working space, where $z = |\LZ(w)|$.
\end{lemma}

\begin{proof}
  We use $\TST{w}{t}$ and the reversed de Bruijn graph of order $t$.
  We associate each leaf of the tree representing a $t$-gram
  with its lexicographical rank among all leaves in the tree.
    
  Let $r$ be the graph node which represents $w[n] = \$$.
  We simply traverse the graph while scanning
  the input string $w$ from right to left.
  For each $1 \leq i \leq n$,
  this gives us the graph node representing $b_i = w[i..i+t-1]$
  and hence the corresponding leaf of $\TST{w}{t}$.

  $\TST{w}{t}$ and the reversed de Bruijn graph 
  can be constructed in $O(n \log \sigma)$ time with $O(zt)$ working space.
  The ranks of the leaves in $\TST{w}{t}$ can be easily computed
  in $O(zt)$ time by a standard tree traversal.
  Traversing the reversed de Bruijn graph takes $O(n \log \sigma)$ time.
  Hence the lemma holds.
\end{proof}

For each $i \in S(t)$, let $r_i$ be the rank
of the $t$-block $b_i = w[i..i+t-1]$ computed by any of the algorithms above.
Clearly $r_i \in [1..n]$.
For simplicity, assume $\sqrt{t}$ is an integer.
For each position $i \in D$ (where $D$ is the underlying $t$-difference cover),
let $\#_{i} = r_{i} r_{i+t} \cdots r_{i+m_it}$,
where $m_i = \frac{n-i+1}{\sqrt{t}}-1$.
We create a string $\code(w) = \#_1 \$_1 \cdots \#_k\$_k$
of length $|S(t)| = O(\frac{n}{\sqrt{t}})$.
Since each $\#_i$ is a string over the integer alphabet $[1..S(t)] \subset [1..n]$
and $|D| = O(\sqrt{t})$, we can regard $\code(w)$ as a string
over an integer alphabet of size $O(n)$.
Then, we build the suffix array,
the inverse suffix array, the LCP array~\cite{manber93:_suffix_array}
of $\code(w)$
and an range minimum query (RMQ) data structure~\cite{bender00:_lca_probl_revis}
for the LCP array.
For any position $i \in S(t)$ on the original string $w$,
we can compute its corresponding position $i'$ on $\code(w)$ as
$i' =  |\#_1 \$_1 \#_2 \$_2 \cdots \#_{x-1}\$_{x-1}| + \frac{i-x}{t} + 1$
where $x = i \bmod t$.
Now, $\LongLCE_t(i, j)$ query for two positions $i, j \in S(t)$
on the original string $w$ reduces to an LCE query
for the corresponding positions on $\code(w)$,
which can be answered in $O(1)$ time using an RMQ on the LCP array.
All these arrays and the RMQ data structure
can be built in $O(\frac{n}{\sqrt{t}})$ time~\cite{Karkkainen:2006je,kasai01:_linear_time_longes_common_prefix,bender00:_lca_probl_revis}.

\begin{theorem} \label{theo:long_LCE}
For any string of length $n$ and integer $1 \le t \le n$, 
a data structure of size $O(\frac{n}{\sqrt{t}})$ can be constructed 
in $O(n \log \sigma)$ time
using $O(zt + \frac{n}{t})$ working space
such that subsequent $\LongLCE_t(i,j)$ queries for any $1 \le i,j \le n$ can be answered in $O(1)$ time, 
where $z = |\LZ(w)|$.
\end{theorem}

\begin{proof}
  We need $O(\frac{n}{\sqrt{t}})$ working space
  for the encoded string $\code(w)$ and its suffix array plus LCP array
  enhanced with an RMQ data structure.
  Then the theorem follows from
  Theorem~\ref{theo:short_LCE},
  and Lemma~\ref{lem:sort_t-blocks_1} or Lemma~\ref{lem:sort_t-blocks_2}.
\end{proof}

\subsection{Main result and variants}

In what follows, let $w$ be an input string of length $n$ and $z = |\LZ(w)|$.
By Theorem~\ref{theo:short_LCE}  and Theorem~\ref{theo:long_LCE}
shown in the previous subsections,
we obtain the main theorem of this paper: 
\begin{theorem} \label{theo:main_theorem}
For any integer $1 \le t \le n$, 
an encoding LCE data structure of size $O(zt + \frac{n}{\sqrt{t}})$
can be constructed 
in $O(n \log \sigma)$ time with $O(zt + \frac{n}{\sqrt{t}})$ working space
such that subsequent $\LCE(i,j)$ query for any $1 \le i,j \le n$ can be answered in $O(1)$ time.
\end{theorem}

We can also obtain the following variants of our LCE data structure.
\begin{corollary} \label{coro:optimal_size}
For any integer $1 \le t \le n$, 
an encoding LCE data structure of size $O(z^{\frac{1}{3}} n^{\frac{2}{3}})$
can be constructed 
in $O(n \log \sigma \log n)$ time with $O(z^{\frac{1}{3}} n^{\frac{2}{3}})$ working space
such that subsequent $\LCE(i,j)$ query for any $1 \le i,j \le n$ can be answered in $O(1)$ time.
\end{corollary}

\begin{proof}
The LCE data structure of Theorem~\ref{theo:main_theorem}
for $t = (\frac{n}{z})^{\frac{2}{3}} < n$ takes
$O(z^{\frac{1}{3}} n^{\frac{2}{3}})$ space.
Since we do not compute $z$, we are not able to compute
the exact value of $(\frac{n}{z})^{\frac{2}{3}}$.
However, by performing doubling-then-binary searches for $t$
and comparing the actual size of $\TST{w}{t}$ and $\lceil \frac{n}{\sqrt{t}} \rceil$ for each tested $t$,
we can obtain the LCE data structure of optimal size,
which can take at most $O(z^{\frac{1}{3}} n^{\frac{2}{3}})$ space.
This takes $O(n \log \sigma \log n)$ total time and
uses $O(z^{\frac{1}{3}} n^{\frac{2}{3}})$ total working space.
\end{proof}

\begin{corollary} \label{coro:little_oh_nlogn}
For alphabets of size $\sigma \leq 2^{o(\log n)}$,
an encoding LCE data structure of size $o(n \log n)$ \emph{bits}
can be constructed 
in $o(n \log^2 n)$ time with $o(n \log n)$ \emph{bits} of working space
such that subsequent $\LCE(i,j)$ query for any $1 \le i,j \le n$ can be answered in $O(1)$ time.
\end{corollary}

\begin{proof}
By plugging the well-known fact that $z = O(n / \log_\sigma n)$
into the result of Corollary~\ref{coro:optimal_size},
we get $O(n / (\log_\sigma n)^{\frac{1}{3}})$ for the space bound.
Thus our data structure can be stored in
$\mathcal{S}(n) = O(n (\log n)^{\frac{2}{3}} (\log \sigma)^{\frac{1}{3}})$
\emph{bits} of space in the transdichotomous word RAM~\cite{FredmanW93}
with machine word size $\omega = \Theta(\log n)$.
Hence, for alphabets of size $\sigma \leq 2^{o(\log n)}$,
we obtain an LCE data structure with the claimed bounds.
\end{proof}

We can also obtain a new time-space trade-off LCE data structure.
Observe that using the data structure of Theorem~\ref{theo:short_LCE}
for $1 \leq d \leq n$,
we can answer $\ShortLCE_{t}$ queries for any $1 \leq t \leq n$ in
$O(\max \{ 1, \frac{t}{d} \})$ time.
Hence the following theorem holds.

\begin{theorem} \label{theo:generic_tradeoff}
	For any integers $1 \le t' \leq t \le n$, 
	a data structure of size $O(zt' + \frac{n}{\sqrt{t}} + \frac{n}{t'})$ can be constructed 
	in $O(n \log \sigma)$ time with $O(zt + \frac{n}{\sqrt{t}} + \frac{n}{t'})$ working space
	such that subsequent $\LCE(i,j)$ query for any $1 \le i,j \le n$ can be answered in $O(\frac{t}{t'})$ time.
\end{theorem}
Theorem~\ref{theo:generic_tradeoff} implies the following:
(1) By setting $t' = t$, we obtain Theorem~\ref{theo:main_theorem}.
Moreover, by choosing also $t \leftarrow (\frac{n}{z})^{2/3}$, 
we obtain a data structure of size $O(z^{1/3}n^{2/3})$ answering LCE queries in constant time, which coincides with Corollary~\ref{coro:optimal_size}.
This is the smallest data structure 
among the fastest data structures with two parameters $t$ and $t'$.
(2) By setting $t' = \sqrt{t}$ and for $t = n/z$,
we get a data structure of size $O(\sqrt{nz})$ answering LCE queries in $O(\sqrt{\frac{n}{z}})$ time. 
This is the fastest data structure among the smallest data structures
with two parameters $t$ and $t'$. 
Note that when we do not know $z$, this data structure of at most $O(\sqrt{nz})$ space can be constructed 
in $O(n \log \sigma \log n)$ preprocessing time and 
$O(\sqrt{nz})$ working space as in Corollary~\ref{coro:optimal_size}.
Although the parameters cannot be arbitrarily chosen,
the space-query time product obtained here
is optimal with fastest construction to date.

Moreover, we can reduce the $zt$ term in the working space of Theorem~\ref{theo:generic_tradeoff} to $zt'$ by 
increasing the preprocessing time. 
The bottle neck of the working space is in sorting $t$-blocks, i.e.,
Lemma~\ref{lem:sort_t-blocks_1} or Lemma~\ref{lem:sort_t-blocks_2}. 
Since any two $t$-blocks can be compared in $O(\frac{t}{t'})$ time 
using $O(\frac{t}{t'})$ $\ShortLCE_{t'}$ queries, 
we can get the following theorem using any suitable comparison-based sorting algorithm 
instead of Lemma~\ref{lem:sort_t-blocks_1} or Lemma~\ref{lem:sort_t-blocks_2}. 
\begin{theorem} \label{theo:tradeoff2}
	We can construct the data structure of Theorem~\ref{theo:generic_tradeoff} in 
	$O(\frac{n}{t'} \log \frac{n}{t} + n \log \sigma)$ time 
	and $O(zt' + \frac{n}{\sqrt{t}} + \frac{n}{t'})$ working space.
\end{theorem}

\section{Lower bounds vs upper bounds for the LCE problem} \label{sec:lowerbound}

Let $\mathcal{T}(n)$ and $\mathcal{S}(n)$ respectively denote
the query time and data structure size (in \emph{bits})
of an arbitrary LCE data structure for an input string of length $n$.

Brodal et al.~\cite{BrodalDR12} showed that
in the non-uniform cell probe model,
any indexing RMQ data structure for a string of length $n$
which uses $\frac{n}{t}$ bits of additional space
for any $1 \leq t \leq n$ must take $\Omega(t)$ query time
(i.e., $\Omega(t)$ cell probes).
Their proof assumes that each character in the string is stored
in a separate cell, and counted the minimum number of character
accesses required to answer an RMQ.
Although their proof uses a binary string of length $n$
where each character takes only a single bit,
the above assumption is valid
in a commonly accepted case that the underlying alphabet size is $2^{\omega}$,
where $\omega$ denotes the size of each cell (i.e. machine word).
Then, Bille et al.~\cite{bille14:_time} showed that
RMQ queries on any binary string of length $n$ can be reduced to
LCE queries on the same binary string,
with $\Theta(\log n)$ additional bits of space.
This implies that, again assuming that each character is stored in a separate cell, any indexing LCE data structure 
for a binary string of length $n$
which uses $\mathcal{S}(n) = \frac{n}{t} + \Theta(\log n)$ additional bits of space
must take $\mathcal{T}(n) = \Omega(t)$ query time,
for parameter $1 \leq t \leq \frac{n}{\log n}$.
Recently, Kosolobov~\cite{Kosolobov16} showed another result
on time-space product trade-off lower bound in the non-uniform cell probe model,
which can be formalized as follows:
\begin{theorem}[\cite{Kosolobov16}]
In the non-uniform cell probe model where each character is stored
in a separate cell, 
for any $\mathcal{S}(n)$, there exists $\sigma = 2^{\Omega(\mathcal{S}(n)/n)}$
such that for any indexing LCE data structure for a string over
the alphabet $\Sigma = \{ 1, \ldots, \sigma\}$,
which takes $\mathcal{S}(n)$ bits of space and answers LCE queries in $\mathcal{T}(n)$ time
(i.e., with $\mathcal{T}(n)$ character accesses or cell probes),
$\mathcal{T}(n) \mathcal{S}(n) = \Omega(n\log n)$ holds.
\end{theorem}
The lower bound by Kosolobov is optimal for the considered range
of the alphabet size $\sigma = 2^{\Omega(\mathcal{S}(n)/n)}$,
since the data structure of Bille et al.~\cite{bille15:_longes_common_exten_sublin_space} achieves $\mathcal{T}(n)\mathcal{S}(n) = O(n \log n)$.

Interestingly,
using our encoding LCE data structure proposed
in Section~\ref{sec:LCE_data_structure},
the above lower bounds can be ``surpassed'' in some cases.
For highly compressible strings
where $zt$ is dominated by $\frac{n}{\sqrt{t}}$,
our LCE data structure of Theorem~\ref{theo:main_theorem}
takes $O(\frac{n \log n}{t})$ \emph{bits} of space for $1 \leq t \leq n$
with machine word of size $\omega = \Theta(\log n)$.
Hence, for parameter $1 \leq t' \leq \frac{\sqrt{n}}{\log n}$
we get $\mathcal{S}(n)  = O(\frac{n}{t'})$.
Since our data structure of Theorem~\ref{theo:main_theorem}
always achieves $\mathcal{T}(n) = O(1)$ for any parameter setting, 
we break Bille et al.'s lower bound for highly repetitive strings.
Notice also that our LCE data structure of Corollary~\ref{coro:little_oh_nlogn}
achieves $\mathcal{T}(n)\mathcal{S}(n) = o(n \log n)$ for alphabet size $\sigma \leq 2^{o(\log n)}$,
which ``surpasses'' Kosolobov's lower bound.
This implies that the alphabet size 
$\sigma = 2^{\Omega(\mathcal{S}(n)/n)}$ is important
for his lower bound to hold.

Kosolobov~\cite{Kosolobov16} did suggest a possibility to overcome his lower bound
when $\sigma$ is small, and the input string can be {\em packed}, where $\log_\sigma n$
characters can occupy a memory cell, allowing the algorithm to read $\log_\sigma n$
characters with one memory access. We show below that this is also possible.
An input string of length $n$ can be considered as a bit string of length $n\log\sigma$. 
Let $t = \log n$, and first consider the $\ShortLCE_{\log n}$ queries on the bit string.
When the original string is available in a packed representation,
the longest common prefix of two substrings strings of length $\log n$ bits can be computed in constant time using no extra space
using bit operations, namely, by taking the bitwise exclusive or (XOR) and computing the position of the most significant set bit (msb), or without msb, by multiple lookups on a table of total size $o(n)$ bits.
Next, consider the $\LongLCE_{\log n}$ queries on the bit string.
By simply using the same data structure as described in Section~\ref{sec:long} for the bit string of length $n\log\sigma$, we can answer $\LongLCE_{\log n}$ queries in constant time using a data structure of size
$O(\frac{n\log\sigma}{\sqrt{\log n}}\log(n\log\sigma)) = O(n \sqrt{\log n} \log \sigma)$ bits.
Using the two queries, we can answer an LCE query for arbitrary positions $i,j$ of the original string in constant time with $\lfloor (\LCE(i\cdot \log\sigma, j\cdot \log\sigma))/\log\sigma\rfloor$.
Since the size of the data structure is $\mathcal{S}(n) = O(n \sqrt{\log n} \log \sigma)$ bits,
we obtain $\mathcal{T}(n)\mathcal{S}(n) = o(n \log n)$ for $\sigma \leq 2^{o(\sqrt{\log n})}$.
Our encoding LCE data structure based on truncated suffix trees is superior for
larger $\sigma$, and also when the input string is highly repetitive
and compressible since it does not require the original string.

\section{Conclusions and open questions} \label{sec:conclusions}

In this paper, we presented an encoding LCE data structure
which uses $O(zt + \frac{n}{\sqrt{t}})$ words of space
and answers in LCE queries in $O(1)$ time, for parameter $1 \leq t \leq \sqrt{n}$.
This data structure can be constructed in $O(n \log \sigma)$ time with
$O(zt + \frac{n}{\sqrt{t}})$ working space.
Using the fact that $z = O(n/\log_\sigma n)$ and suitably choosing $t$,
our method achieves the first $O(1)$-time sub-linear space LCE data structure
for alphabets of size $\sigma \leq 2^{o(\log n)}$.

An interesting open question is whether we can improve the total space requirement to $O(zt + \frac{n}{t})$.
The bottle neck is the $\LongLCE_t$ data structure that uses $O(zt + \frac{n}{\sqrt{t}})$ space.
Another open question is whether we can compute
the size $z$ of the Lempel-Ziv 77 factorization 
in $O(n \log \sigma)$ time with sub-linear working space.
This is motivated for computing the value of $t$
which optimizes our space bound $O(zt + \frac{n}{t})$.
A little has been done in this line of research:
Nishimoto et al.~\cite{NishimotoIIBT16PSC} showed how to compute
the Lempel-Ziv 77 factorization
in $O(n \polylog(n))$ time with $O(z \log n \log^* n)$ working space.
Fischer et al.~\cite{FGGK15} showed an algorithm which computes
an approximation of the Lempel-Ziv 77 factorization
of size $(1+\epsilon)z$ in $O(\frac{1}{\epsilon} n \log n)$ time
with $O(z)$ working space, for any $0 < \epsilon \leq 1$.

Another direction of further research is to give a tighter upper bound
for the size of the $t$-truncated suffix trees than $zt$.
We observed that there exists a string of length $n$
for which $zt$ is greater by a factor of $\sqrt{n}$ than
the actual size of the $t$-truncated suffix tree for some $t$.

\section*{Acknowledgments}
We thank Dmitry Kosolobov for explaining his work~\cite{Kosolobov16} to us.

\clearpage
\bibliography{ref}
\bibliographystyle{plain}
\end{document}